\documentclass[%
 twocolumn,
 superscriptaddress,
 pre,
 floatfix,
 aps
]{revtex4-2} 

\usepackage[utf8]{inputenc}
\usepackage[ruled,vlined,linesnumbered]{algorithm2e}
\usepackage{amsmath}
\usepackage{amsthm}
\usepackage{amssymb}
\usepackage{bookmark}
\usepackage{hyperref} 
\usepackage{xcolor}
\hypersetup{
    colorlinks=true,
    linkcolor={red!50!black},
    filecolor={green!50!black},
    citecolor={green!50!black},
    urlcolor={blue!80!black},
}
\usepackage[capitalise]{cleveref}
\crefname{algorithm}{Alg.}{Algs.}
\Crefname{algorithm}{Algorithm}{Algorithms}
\crefname{appendix}{App.}{App.}
\Crefname{appendix}{Appendix}{Appendices}
\crefname{corollary}{Corol.}{Corolls.}
\Crefname{corollary}{Corollary}{Corollaries}
\crefname{conjecture}{Conjecture}{Conjectures}
\Crefname{conjecture}{Conjecture}{Conjectures}
\crefformat{conjecture}{Conjecture~#2#1#3}
\crefmultiformat{conjecture}{Conjectures~#2#1#3}{\ and~#2#1#3}{, #2#1#3}{\ and~#2#1#3}
\crefname{definition}{Def.}{Defs.}
\Crefname{definition}{Definition}{Definition}
\crefname{figure}{Fig.}{Figs.}
\Crefname{figure}{Figure}{Figures}
\crefname{lemma}{Lemma}{Lemmas}
\Crefname{lemma}{Lemma}{Lemmas}
\crefname{proposition}{Prop.}{Props.}
\Crefname{proposition}{Proposition}{Propositions}
\Crefname{section}{Section}{Sections}
\crefname{section}{Sect.}{Sect.}
\crefname{subsection}{Sect.}{Sect.}
\Crefname{subsection}{Section}{Sections}
\crefname{subsubsection}{Sect.}{Sect.}
\Crefname{subsubsection}{Section}{Sections}
\crefname{table}{Table}{Tables}
\Crefname{table}{Table}{Tables}
\crefname{theorem}{Thm.}{Thms.}
\Crefname{theorem}{Theorem}{Theorems}
\usepackage{mathtools} 
\usepackage{MnSymbol}
\PassOptionsToPackage{hyphens}{url} 
\usepackage{scalerel}
\usepackage{footnote}
\makesavenoteenv{tabular}
\makesavenoteenv{table}
\usepackage{xfrac} 
\usepackage{xspace} 
\usepackage{booktabs}
\usepackage{multirow}

\newcommand{\mpara}[1]{\medskip\noindent{\bf #1}}

\DeclarePairedDelimiter\abs{\lvert}{\rvert} 
\makeatletter
\let\oldabs\abs%
\def\abs{\@ifstar{\oldabs}{\oldabs*}}
\makeatother
\newcommand{\bfc}[1]{\ensuremath{\mathsf{b}\lparen#1\rparen}} 
\newcommand{\bfcn}[2]{\ensuremath{\mathsf{b}_{#1}\lparen#2\rparen}} 
\newcommand{\bfcnn}[3]{\ensuremath{\mathsf{b}_{#1}\lparen#2,#3\rparen}} 
\newcommand{\card}[1]{\ensuremath{\abs{#1}}\xspace} 
\newcommand{\degree}[2]{\ensuremath{\mathsf{d}_{#1}\lparen#2\rparen}\xspace} 
\newcommand{\distrib}{\ensuremath{\pi}\xspace} 
\newcommand{\ensemble}{\ensuremath{\mathcal{H}}\xspace} 
\newcommand{\graphspace}{\ensuremath{\mathcal{G}}\xspace} 
\DeclarePairedDelimiter\mylist{\langle}{\rangle} 
\newcommand{\neigh}[2]{\ensuremath{\Gamma_{#1}\lparen#2\rparen}\xspace} 
\newcommand{\suchthat}{\ensuremath{\mathrel{:}}\xspace} 
\newcommand{\sw}{\ensuremath{\mathsf{sw}}\xspace} 

\newtheorem{theorem}{Theorem}
\newtheorem{definition}{Definition}
\newtheorem{lemma}{Lemma}
\newtheorem{fact}{Fact}

\newenvironment{squishlist}
{\begin{list}{$\bullet$}
 {\setlength{\itemsep}{0pt}
     \setlength{\parsep}{3pt}
     \setlength{\topsep}{3pt}
     \setlength{\partopsep}{0pt}
     \setlength{\leftmargin}{1.5em}
     \setlength{\labelwidth}{1em}
     \setlength{\labelsep}{0.5em} } }
{\end{list}}

\newcommand{\changed}[1]{{\color{black} #1}}
\newcommand{\rev}[1]{{\color{black} #1}}

\begin{document}

\title{\texorpdfstring{An impossibility result for Markov Chain Monte Carlo
sampling\\from micro-canonical bipartite graph ensembles}{An impossibility result for Markov Chain Monte Carlo
sampling from micro-canonical bipartite graph ensembles}}

\author{Giulia Preti}
\email{giulia.preti@centai.eu}
\author{Gianmarco De Francisci Morales}
\email{gdfm@acm.org}
\affiliation{CENTAI, Corso Inghilterra 3, 10138 Turin, Italy}
\author{Matteo Riondato}
\email{mriondato@amherst.edu}
\affiliation{Amherst College, 25 East Drive, Amherst, 01002 Massachusetts, USA}

\begin{abstract}
  Markov Chain Monte Carlo (MCMC) algorithms are
  commonly used to sample from graph ensembles. Two graphs are neighbors in the
  state space if one can be obtained from the other with only a few modifications, e.g., edge rewirings.
  For many common ensembles, e.g., those preserving the degree
  sequences of bipartite graphs, rewiring operations involving
  two edges are sufficient to create a fully-connected state
  space, and they can be performed efficiently.
  We show that, for ensembles of bipartite graphs with fixed
  degree sequences and number of \emph{butterflies} ($k_{2,2}$ bi-cliques),
  there is no universal constant $c$ such that a rewiring of at most $c$ edges at every step is
  sufficient for \emph{any} such ensemble to be fully connected. Our proof
  relies on an explicit construction of a family of pairs of graphs with the
  same degree sequences and number of butterflies, with each pair
  indexed by a natural $c$, and such that any sequence of rewiring operations
  transforming one graph into the other \emph{must} include at least one
  rewiring operation involving at least $c$ edges.
  Whether rewiring these many edges is \emph{sufficient} to guarantee
  the full connectivity of the state space of any such ensemble remains an open
  question. Our result implies the impossibility of developing efficient,
  graph-agnostic, MCMC algorithms for these ensembles, as the necessity to
  rewire an impractically large number of edges may hinder taking a step on the
  state space.
\end{abstract}

\maketitle

\section{Introduction}\label{sec:intro}
\enlargethispage{0.5\baselineskip}

Testing the statistical significance of properties of an observed network is
a fundamental problem in network science~\cite{newman2003structure}. The significance of
the observed value is tested against a null model, an ensemble
$\ensemble=(\graphspace, \changed{\distrib})$ composed of the set $\graphspace$
of possible graphs that can be realized under the null hypothesis and a
probability distribution $\changed{\distrib}$ over $\graphspace$. 
One typically selects
some descriptive characteristics of the observed network, and either defines
$\graphspace$ and $\changed{\distrib}$ in such a way that the \emph{expectations} w.r.t.\
$\changed{\distrib}$ of these characteristics over $\graphspace$ are the same as the
observed \changed{ones} (a.k.a.\ the \emph{canonical} model), or defines $\graphspace$ as the set of all and only the
graphs with \emph{exactly the same} values for the characteristics as the observed
network, and $\changed{\distrib}$ can be any distribution, often the uniform.
Once the null model is defined, one proceeds by sampling several graphs from this
ensemble. These graphs are used to approximate the distribution of the test
statistic of interest under the null hypothesis. By comparing the observed
statistic to this distribution one can compute an empirical $p$-value.

For example, the widely used ``configuration
model''~\citep{fosdick2018configuring} considers the set of graphs with the same
degree sequence as the observed network and the uniform distribution. This
model has been instrumental in determining that clustering, assortativity, and community 
structure in real networks are not solely dependent on node degrees, hence highlighting their
significance~\citep{cimini2019statistical}. However,
the configuration model fails to generate graphs with a local structure similar to
the observed graph~\citep{stanton2012constructing}. Researchers have thus
explored alternative null models that sample from graph families defined by more
complex characteristics of the observed graph, such as joint degree
distribution\rev{~\citep{mahadevan2006systematic,stanton2012constructing,orsini2015quantifying,bassler2015exact}}, core-value
sequence~\citep{van2021random}, and local triangle-count
sequence~\citep{newman2009random}.

In this work, we focus on bipartite graphs, i.e., networks whose nodes 
can be partitioned into two classes such that all edges go from one class to the other.
Formally, a bipartite graph is a tuple $G \doteq (L, R, E)$, where $L$ and $R$ are
disjoint sets of nodes called \emph{left} and \emph{right} nodes, respectively,
and $E \subseteq L \times R$ is a set of edges connecting nodes in $L$ to nodes
in $R$. 
We consider undirected
bipartite graphs, but for ease of presentation, we denote any edge $(u,a)$ so
that $u \in L$ and $a \in R$.
For any vertex $v \in L \cup R$ we denote with $\neigh{G}{v}$ the set of
\emph{neighbors} of $v$, i.e., the vertices to which $v$ is connected by an edge
in $G$, and we define the \emph{degree} $\degree{G}{v}$ of $v$ in
$G$ as $\degree{G}{v} \doteq \card{\neigh{G}{v}}$. Assuming an arbitrary but fixed
labeling $u_1,\dotsc, u_{\card{L}}$ (resp.~$a_1,\dotsc,a_{\card{R}}$) of the
nodes in $L$ (resp.~$R$), the vector $\mylist{\degree{G}{u_1},
\dotsc, \degree{G}{u_{\card{L}}}}$ (resp.~$\mylist{\degree{G}{a_1},
\dotsc, \degree{G}{a_{\card{R}}}}$) is known as the \emph{left (resp.~right)
degree sequence} of $G$.

Bipartite networks occur naturally in many applications: when representing words and
documents~\citep{zha2001bipartite}, items and itemsets~\citep{preti2022alice},
higher-order networks such as hypergraphs and simplicial
complexes~\citep{berge1970certains}, and many more. 
Null models and graph
ensembles can also be defined on bipartite 
graphs~\citep{saracco2015randomizing,preti2022alice,NealCGGSSSUWS23}. For example,
\citet{preti2022alice} introduce a null model that preserves the bipartite
joint adjacency matrix \changed{(i.e., the matrix whose $(i, j)$-th entry is the
number of edges connecting nodes from $L$ with degree $i$ to nodes in $R$ with
degree $j$)}, of an observed network (thus the degree sequences and the number
of \emph{caterpillars}, i.e., paths of length 3), and give Markov Chain Monte
Carlo (MCMC) algorithms to sample from this null model. Null models for
bipartite graphs are also of particular interest because they align with null
models for 0--1 binary matrices~\citep[Ch.\ 6]{ryser1963combinatorial}.
For example, preserving
the degree sequences in a bipartite graph corresponds to preserving the row and column marginals of the corresponding 
bi-adjacency matrix, and several MCMC algorithms have been developed to sample from this
null model~\citep{carstens2015proof,verhelst2008efficient,wang2020fast,kannan1999simple,strona2014fast}.

We consider graph ensembles for which $\graphspace$ is the set of all bipartite
graphs $G \doteq (L, R, E)$ that share the same degree sequences and the same number of
\emph{butterflies}, i.e., $k_{2,2}$ bi-cliques\changed{, defined as follows.}

\begin{definition}[Butterfly]\label{def:butterfly}
  Let $G \doteq (L, R, E)$ be a bipartite graph. Two distinct nodes $u,v \in
  L$ and two distinct nodes $a,b \in R$ \emph{belong} to the \emph{butterfly}
  $A=\{u, v, a, b\}$ in $G$ if and only if $\{(u,a), (u,b), (v,a), (v,b)\} \subseteq E$.
\end{definition}

The following result, whose proof is immediate, gives an expression for the
number of butterflies to which two nodes both belong.

\begin{fact}\label{lem:bfcnn}
  Let $G \doteq (L, R, E)$ be a bipartite graph, and let $u$ and $v$ be
  distinct nodes in $L$. The number $\bfcnn{G}{u}{v}$ of butterflies in $G$ to
  which both $u$ and $v$ belong is
  \[
    \bfcnn{G}{u}{v} = \binom{\card{\neigh{G}{u} \cap \neigh{G}{v}}}{2}\,,
  \]
  where we assume $\binom{0}{2}=\binom{1}{2}=0$. A similar result holds for any
  two distinct nodes in $R$.
\end{fact}

For $u \in L$, we denote with $\bfcn{G}{u}$ the number of butterflies in $G$
to which $u$ belongs. It holds
\begin{equation}\label{eq:bfcn}
  \bfcn{G}{u} = \sum_{\substack{v \in L\\v \neq u}} \bfcnn{G}{u}{v} \,.
\end{equation}

The total number $\bfc{G}$ of butterflies in $G$ is then
\begin{equation}\label{eq:bfc}
  \bfc{G} \doteq \frac{1}{2} \sum_{u \in L} \bfcn{G}{u} \,.
\end{equation}

The butterfly, being the smallest
complete subgraph in a bipartite graph, is the most basic building block for composing more complex structures, analogous to the triangle in unipartite graphs.
Consequently, preserving the number
of butterflies emerges as a natural choice when defining null
models that retain more graph properties beyond just the degree sequences.
This concept finds applications in studying, e.g.,
clustering patterns~\citep{nishimura2017swap}.

MCMC methods are a popular approach to sample from an ensemble
$\ensemble=(\graphspace, \changed{\distrib})$.
They define a suitable Markov chain on the
space $\graphspace$ of all possible graphs, such that, \changed{after a sufficient burn-in period}, the state of the
Markov chain is approximately distributed according to $\changed{\distrib}$.
The correctness of this process requires the Markov
chain to be finite, irreducible, and aperiodic~\citep{fosdick2018configuring}.
Efficient sampling requires not only that the Markov chain has a fast mixing
time, but also that the space can be explored quickly, i.e., that
obtaining a neighbor from the current state is efficient.
The double edge swap technique, also known as degree-preserving
rewiring~\citep{cafieri2010loops}, checkerboard
swap~\citep{artzy2005generating}, or tetrad~\citep{verhelst2008efficient}, is a simple yet fundamental randomization technique used to generate a new graph with the same
degree sequence as a given graph. 
Its efficiency stems from the fact that it involves the rewiring of a \emph{small} number of edges.
In bipartite graphs, the most basic rewiring technique is known as the \emph{bipartite swap operation} (BSO).

\begin{definition}[BSO]\label{def:bso}
Let $G \doteq (L, R, E)$ be a bipartite graph and $u \neq v \in L$, $a \neq b \in
R$ such that $(u,a), (v,b) \in E$ and $(u,b), (v,a) \notin E$. 
The BSO involving $(u,a)$ and $(v,b)$ removes $(u,a)$ and $(v,b)$ from $E$, 
and adds $(u,b)$ and $(v,a)$ to $E$.
The resulting bipartite graph $G'=(L, R, (E \setminus \{(u,a), (v,b)\}) \cup \{(u,b),
(v,a)\})$ has the same left and right degree sequence of $G$.
\end{definition}

A more sophisticated operation is the \emph{$q$-edge bipartite swap operation}
($q$-BSO), which \changed{may involve the simultaneous swapping of multiple
edges, potentially between a large set of nodes}, similar
to the $q$-switch operation defined by \citet{tabourier2011generating}.

\begin{definition}[$q$-BSO]\label{def:qbso} 
Let $G \doteq (L, R, E)$ be a bipartite graph and $q \in \mathbb{N}^+$. 
A $q$-BSO is a pair $\sw^q \doteq (S, \sigma)$ with $S = \mylist{e_1, \dotsc, e_q}$ being a vector of $q$ \emph{distinct} edges $e_i \doteq (u_i,a_i) \in E$, and $\sigma$ being a \changed{derangement of $[q]$, i.e., a permutation of $[q]$ with no element in its original position,} s.t.\ $(u_j, a_{\sigma(j)}) \notin E$ for each $j \in [1,q]$.
Replacing each $e_j$ with $(u_j, a_{\sigma(j)})$ generates a bipartite graph $G'=\left(L, R, (E \setminus S) \cup \left\{\left(u_j, a_{\sigma(j)}\right) \text{ for } j \in [1, q]\right\}\right)$ with the same
  left and right degree sequence as $G$.
\end{definition}

According to this definition, a BSO involving $(u,a)$ and $(v,b)$ can be seen as
the $2$-BSO $(\mylist{(u,a),(v,b)}, (2\ 1))$. \changed{Algorithms such as
Verhelst's~\citep{verhelst2008efficient} and Curveball~\citep{strona2014fast}
aim to speed up the sampling from the ensemble of bipartite graphs with
fixed degree sequences. They execute multiple BSO operations at each
step by selecting nodes $u$ and $v$ from $L$ (or $R$) and exchanging multiple
edges originating from $u$ with edges originating from $v$. Conversely, a
$q$-BSO may involve the simultaneous swapping of multiple edges originating from
\emph{multiple} source nodes. Thus, the moves considered by Curveball and
Verhelst's can be expressed as $q$-BSOs, but $q$-BSOs are more expressive, in
the sense that there are $q$-BSOs that do not correspond to possible moves for
these algorithms.}

\section{Connectivity of the state space}\label{sec:conn}

A key requirement to use an MCMC method for sampling from a graph ensemble is
that the state space, where each state corresponds to a graph in the ensemble,
is strongly connected\changed{, i.e., for any two states $G'$ and $G''$ there
is a sequence $\mylist{\rho_1, \rho_2, \dotsc, \rho_{\ell}}$ of graph-transforming operations 
for some $\ell$ (which may depend on the chosen $G'$ and $G'')$, such
that $\rho_1$ transforms $G'$ into some $G_1$ that belongs to the state space,
$\rho_i$ for $1 < i < \ell$ transforms $G_i$ into $G_{i+1}$ that also belongs
to the state space, and $\rho_\ell$ transforms $G_{\ell-1}$ into $G''$. In
other words, a class $\mathcal{C}$ of graph-transforming operations defines a
neighborhood structure of the state space as follows: given any $G$ in the state
space, a neighbor of $G$ is any state that can be obtained by applying a single operation from $\mathcal{C}$, provided that the operation is applicable to
$G$. With this neighborhood structure, the state space is strongly
connected if there is a path from any state to any other state.}

We can immediately see that the state space $\graphspace$ we consider is
\changed{strongly} connected by
sequences of $q$-BSOs when $q$ is large enough, for any left and right degree sequences, and any number of butterflies (see also \citep[Sect.\
3.2.2]{tabourier2011generating}).
In fact, there is always a \changed{$\card{E' \setminus E''}$-BSO} that transforms any bipartite graph $G' \doteq (L,R,E')$ into another bipartite graph $G'' \doteq (L,R,E'')$ with the same left and right degree sequences, and number of butterflies (see Supplementary Material~\citep{PretiDFMR24supp} for details).
While this \changed{fact} ensures the strong connectivity of the state space \changed{via the union of all $q$-BSOs for $q=2,\dotsc,\card{E'}$}, it has little practical relevance, as we now explain. If we
use \changed{all these $q$-BSOs} to define the neighborhood structure of the state space,
the resulting space would be a \emph{complete} graph, i.e., a clique.
Consequently, \changed{drawing, according to any distribution,} a neighbor
\changed{of} a given state would \changed{require a procedure to build an
entirely new bipartite graph with the same degree sequences and the same number
of butterflies \emph{from scratch}. Developing such a procedure seems even
harder than the problem we are attempting to solve, in the same way as \rev{devising algorithms for} building
a bipartite graph with prescribed degree sequences from scratch
\citep{chen2005sequential,Miller_2013,holmes1996uniform,miklos2004randomization,snijders1991enumeration}
is much harder than \rev{devising algorithms for} sampling such a graph using MCMC approaches starting from an
existing one
\citep{verhelst2008efficient,strona2014fast,carstens2015proof,wang2020fast,kannan1999simple}.}

The correct question to ask is therefore the following: is there a \emph{fixed,
universal, constant} $q^*$ such that, for any left and right degree sequences,
and any number of butterflies, any two bipartite graphs with those
left and right degree sequences, and that number of butterflies, are connected
by sequences $\langle \sw^{p_1}_1 = (S_1, \sigma_1), \dotsc, \sw^{p_z}_z =
(S_z,\sigma_z)\rangle$ of $p_i$-BSOs, with $p_i \le q^*$, $i=1,\dotsc,z$, where
$z$ may depend on the two graphs?
\changed{By ``universal'', we mean a quantity that in no way depends on properties
  of the ensemble $(\mathcal{G},\pi)$, including properties of the observed
network.}

Asking this question is reasonable: it is known that $q^*=2$ when one is only interested in preserving the
degree sequences~\citep[Ch.\ 6]{ryser1963combinatorial}, and it is also known that $q^*=2$ for the case of preserving
the degree sequences and the number of paths of length 3 (a.k.a.\
caterpillars), in which case the rewiring operations slightly differ from the
traditional BSOs~\citep{preti2022alice,czabarka2015realizations}.

Additionally, we would like $q^*$ to be small, because sampling a $q$-BSO
is not necessarily efficient: the na\"{\i}ve approach of independently sampling
$q$ edges and then verifying whether they form a valid $q$-BSO has an
increasing probability of failure as $q$
increases~\citep{tabourier2011generating}. 
As a result, the Markov chain would exhibit a high probability of staying in the same state for many consecutive steps,
greatly increasing the mixing time.

For unipartite graphs, it has been proved that $q^*=2$ is not always 
sufficient to ensure strong connectivity of spaces of graphs that share more complex properties~\citep{tabourier2011generating, czabarka2015realizations}.
In this work, we demonstrate the nonexistence of a \emph{fixed,
universal, constant} $q^*$ for the ensemble of bipartite graphs with the same left and right degree sequences and the same number of butterflies.

Let us give some intuition with
an example, which shows that it cannot be $q^* < 4$.
\Cref{fig:counter} shows two bipartite graphs $G_1$ (upper) and $G_2$ (lower) with the same left and right
degree sequences, and the same number of butterflies $\bfc{G_1} = \bfc{G_2} = 10$.
There is no sequence of $q$-BSO\changed{s} for $q < 4$ that, when applied to $G_1$,
generates a graph isomorphic to $G_2$: any $q$-BSO for $q < 4$ applied to $G_1$ either generates a graph with a different number of butterflies, or generates a
graph isomorphic to $G_1$. On the other hand, the $4$-BSO
$([(x_1,y_5), (x_5,y_1), (x_6,y_{10}), (x_{10},y_6)], \sigma)$ with $\sigma(1) =
3$, $\sigma(2) = 4$, $\sigma(3) = 1$, and $\sigma(4) = 2$ ensures that the two
butterflies $\{x_1,x_5,y_1,y_5\}$ and $\{x_6,x_{10},y_6,y_{10}\}$ disappear,
while the two new butterflies $\{x_1,x_{10},y_1,y_{10}\}$ and $\{x_5, x_6, y_5,
y_6\}$ appear, hence preserving the total count $\bfc{G_1}$.

\begin{figure}[thb]
  \centering
  \includegraphics[width=\columnwidth]{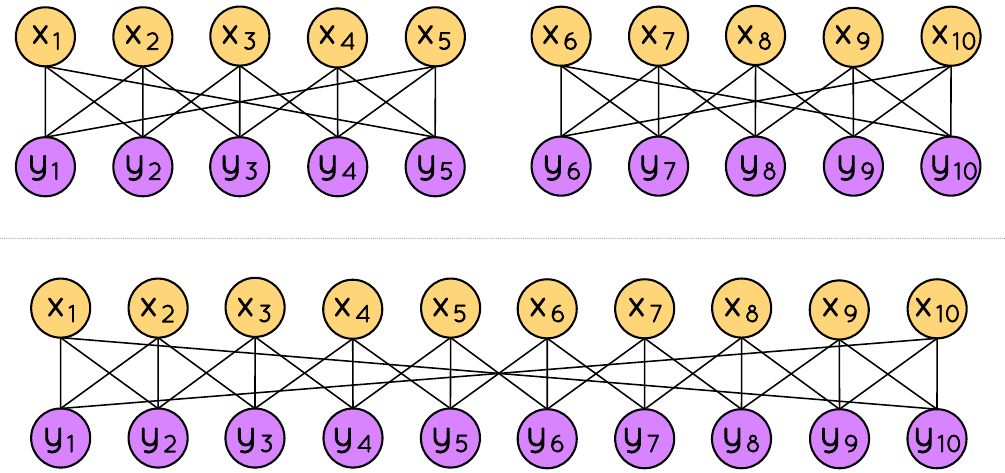}
  \caption{Graphs that are not connected by $q$-BSOs for $q < 4$.}\label{fig:counter}
\end{figure}

Our main result is the following theorem (proof in Supplementary
Material~\citep{PretiDFMR24supp}).

\begin{theorem}\label{thm:main}
  For any $\bar{q} \in \mathbb{N}$ with $\bar{q} > 1$, there exist two non-isomorphic
  bipartite graphs $G_\mathrm{b}$ and $G_\mathrm{e}$ with the same left and
  right degree sequences, and $\bfc{G_\mathrm{b}} = \bfc{G_\mathrm{e}}$, such
  that for any sequence $\langle \sw^{p_1}_1 =
  (S_1, \sigma_1), \dotsc, \sw^{p_z}_z = (S_z,\sigma_z)\rangle$ of
  $p_i$-BSOs with $p_i \in \mathbb{N}^+$, $i=1,\dotsc,z$, that
  transforms $G_\mathrm{b}$ into $G_\mathrm{e}$, there exists $\ell \in
  \{1,\dotsc,z\}$ with $p_\ell \ge \bar{q}$.
\end{theorem}

Our proof consists of two parts. First, we construct two bipartite graphs
$G_\mathrm{b}$ and $G_\mathrm{e}$ with the same left and right degree sequences
(which will depend on $\bar{q}$ as the second largest left degree will be
greater than $\bar{q}$), and the same number of butterflies. Second, we
demonstrate that any sequence of $q$-BSOs applied to $G_\mathrm{b}$ to obtain a
graph isomorphic to $G_\mathrm{e}$ must involve at least one $q$-BSO for $q >
\bar{q}$. Since $\bar{q}$ can be arbitrarily large, a \emph{universal constant}
$q^*$ as above cannot exist.

This theorem proves that it is impossible to design efficient MCMC algorithms
that sample from ensembles $\ensemble=(\graphspace,\changed{\distrib})$ of bipartite graphs with the same
degree sequences and the same number of butterflies, because the state space is
not \changed{strongly} connected by edge swap operations that involve only up to a fixed,
universal, number of edges, as is instead the case for simpler null models.
Rather, the minimum number of edges that \emph{must} be involved depends on
properties of the state space $\mathcal{G}$\changed{, not just of the observed
network}. These may not be easily computable,
as they may not depend \emph{just} on the observed network, if any.

This result has profound implications for the design of network null models and
for network science in general. If it is unfeasible to preserve the occurrences of the simplest 
building block of bipartite graphs (the butterfly), it \changed{becomes
unfeasible to preserve larger structures}.
When dealing with bipartite graphs and complex observed characteristics, ensembles with ``soft'' constraints, where the constraints are retained on average over all the graphs sampled from the ensemble~\cite{robins2007introduction},
might be the only viable option.

\begin{figure}[!t]
  \centering
  \includegraphics[width=\columnwidth]{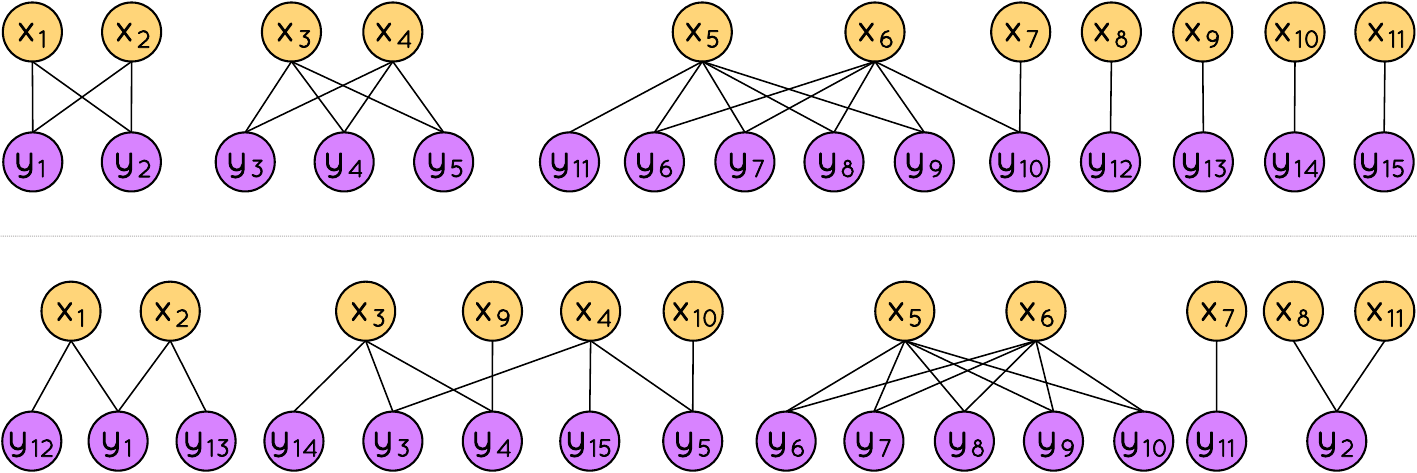}
  \caption{Bipartite graphs generated by our algorithm (see Supplementary
  Material~\citep{PretiDFMR24supp}) for $s = 2$ and $t = 3$.}\label{fig:graphs}
\end{figure}

The algorithm to construct the graphs $G_\mathrm{b}$ and
$G_\mathrm{e}$ is delineated in the Supplementary
Material~\citep{PretiDFMR24supp}. 
We now describe the main characteristics of such graphs.

Let $s$ and $t$ be two naturals with $s > t \ge 2$ and $2 (s - 1) > \bar{q}$. 
We define $n \doteq \binom{s}{2} + \binom{t}{2}$ and $a \doteq ((s+1)\
\text{mod}\ 2) + ((t+1)\ \text{mod}\ 2)$.
The graphs $G_\mathrm{b}$ and $G_\mathrm{e}$ output by the algorithm with
inputs $s$ and $t$ have the following properties:
\begin{enumerate}
  \item $G_\mathrm{b}$ and $G_\mathrm{e}$ have $7+a$ left nodes (denoted with the letter $x$) and $s+t+n+2+a$ right nodes (denoted with the letter $y$); 
  \item $G_\mathrm{b}$ and $G_\mathrm{e}$ have the same left and right degree
    sequences. In particular, $x_1$ and $x_2$ have degree $s$, $x_3$ and $x_4$
    have degree $t$, $x_5$ and $x_6$ have degree $n+1$, $y_i$ has degree 2 for $1 \le i \le s+t+n+1$, and all the other left and right nodes have degree $1$;
  \item $\bfc{G_\mathrm{b}} = \bfc{G_\mathrm{e}} = n + \binom{n}{2}$;
  \item $\card{\neigh{G_\mathrm{e}}{x_5} \cap \neigh{G_\mathrm{e}}{x_6}} = n+1$, which,
  with the previous point, implies that $x_5$ and $x_6$ belong to \emph{all}
  butterflies in $G_\mathrm{e}$, and every other node $y \in R$ belongs to no
  butterfly;
  \item $\card{\neigh{G_\mathrm{b}}{x_5} \cap \neigh{G_\mathrm{b}}{x_6}} = n$,
      $\card{\neigh{G_\mathrm{b}}{x_1} \cap \neigh{G_\mathrm{b}}{x_2}} = s$,
      $\card{\neigh{G_\mathrm{b}}{x_3} \cap \neigh{G_\mathrm{b}}{x_4}} = t$,
      which implies, with point 2, that $x_5$ and $x_6$ do \emph{not} belong to
      \emph{all} butterflies in $G_\mathrm{b}$.
  \end{enumerate}

\Cref{fig:graphs} shows the bipartite graphs generated for $s = 2$ and $t = 3$.

\section{Discussion}\label{sect:conclusions}

MCMC approaches based on swaps of pairs of edges can efficiently sample graphs
from simple ensembles, such as those including graphs with prescribed degree sequences, or
fixed number of paths of length up to three~\citep{preti2022alice}.
The correctness of these approaches relies on the fact that pairwise edge swaps
create a strongly connected state space. They are efficient because proposing a
neighbor to move to is relatively easy, requiring only to be able to efficiently
sample pairs of edges.

In the case of ensembles preserving more complex properties of the networks,
the \changed{strong} connectivity of the state space may require more than two edges to be
swapped at every step, i.e., performing
$q$-switches~\citep{tabourier2011generating}, or $q$-BSOs, for
some value of $q$.

In this work, we consider the ensemble of bipartite graphs with fixed degree
sequences and fixed number of butterflies ($k_{2,2}$ bi-cliques), for its important role in a 
variety of applications, e.g., investigating clustering patterns~\citep{nishimura2017swap}.
We show that the state space is not \changed{strongly} connected by sequences of $q$-BSOs for any \emph{fixed}, \emph{universal},
constant $q$. In other words, the number of edges to be rewired at each
step is upper bounded by a quantity that depends on properties of the graphs in
the ensemble. This result is in strong contrast with the cases for the space of
bipartite graphs with fixed degree sequences, and for that of bipartite graphs
with fixed degree sequences and fixed number of paths of length three (a.k.a.\
caterpillars), where $q=2$ is sufficient for all
ensembles~\citep{preti2022alice,czabarka2015realizations}.

This discovery has far-reaching implications for network science.
First and foremost, we rule out the possibility of designing \emph{efficient} MCMC
algorithms for sampling from the space of bipartite graphs with fixed degree
sequences and fixed number of butterflies, specifically,
from the micro-canonical ensemble that maintains \changed{these properties}
exactly.
In fact, we demonstrate the necessity of swaps with size dependent on the
characteristics of the graph space \changed{$\mathcal{G}$, not necessarily just
on the observed network. Finding what this size \changed{$q^*_{\mathcal{G}}$ is
may not even be feasible. It may perhaps be possible to develop an efficient
procedure to find this quantity, but then one also needs an efficient procedure
to, at each step of the Markov chain}, generate a $q$-BSO for $q \le
q^*_{\mathcal{G}}$, to propose a neighbor to move to}. Solving both these
algorithmic questions seem challenging. \rev{Moreover, the lower bound $\bar{q}$
to the size of the BSOs needed to connect the two graphs $G_\mathrm{b}$ and
$G_\mathrm{e}$ from \cref{thm:main} gives only a \emph{necessary} condition for
the strong connectivity of the graph space, not a \emph{sufficient} one: we only
know that one of the BSOs to connect these two graphs must contain more than
$\bar{q}$ edges, but not how exactly how many. Even if we knew exactly this
number $\hat{q}$, there may be other pairs of graphs in the same ensemble
(i.e., with the same degree sequences and number of butterflies) such that any
sequence of BSOs connecting these two graphs must have size even greater than
$\hat{q}$.} Therefore the situation might be even more dire than our findings
suggest.

Given that sampling from a null model that preserves the number of butterflies is impractical,
preserving larger structures seems an even more unattainable task.
A butterfly is at the same time the smallest cycle and the smallest non-trivial
bi-clique, hence it is a basic building block of bipartite graphs. Thus, our
findings present a large obstacle to developing efficient algorithms to
sample from more complex ensembles, and therefore to testing network
properties under more descriptive null models.

\rev{What other options are then available, if any? If one wishes to maintain the
number of butterflies as a hard constraint (i.e., to sample from the micro-canonical ensemble),
one potential approach involves avoiding MCMC algorithms and opting for a direct-sampling algorithm like stub matching~\cite{bollobas1980probabilistic}.
However, such algorithms are already limited to small graph instances, for the case of sampling from the space of graphs with the same degree sequence, due to their complexity scaling quadratically or cubically with the number of nodes, depending on the graph density~\cite{erdHos2009simple,kim2012constructing}.
The straightforward application of existing stub-matching techniques may also suffer from generating graphs with a different number of butterflies, thus leading to a high rejection rate.
Thus, we need to explore alternative methodologies or refine existing stub-matching algorithms to better accommodate these more complex constraints.
Finally, implementations of canonical methods such as the Chung-Lu model~\cite{chung2002average} offer a more efficient alternative, albeit at the cost of imposing a \emph{soft} constraint.
Indeed, while the ensemble average aligns precisely with the desired value of each constraint, individual graph instances may lie far from the desired constraints.
The canonical ensemble brings other challenges, including difficulties in generating graphs that closely match the desired expectations for certain degree distributions (\emph{degeneracy} problem)~\cite{squartini2011analytical,squartini2015unbiased,brissette2022limitations,snijders2006new,Vallarano_2021}.}

Overall, our findings represent a strong negative result that the network
science community needs to reckon with. By showing that this research avenue is
not fruitful, we hope to spur alternative and innovative approaches to designing
null models for graphs, and algorithms for sampling from them.

\mpara{Acknowledgments.}
  MR's work was sponsored in part by NSF grants IIS-2006765 \changed{and
  CAREER-2238693.}

\bibliographystyle{apsrev4-2} 
\bibliography{refs}

\appendix

\section{Proofs}\label{sec:app:proofs}

Our main result, Theorem 1, relies on the following lemma.

\begin{lemma}\label{lem:maxbfcn}
  For any bipartite graph $G \doteq (L, R, E)$ with $\degree{G}{a} \le 2$,
  for each $a \in R$, and any $u \in L$, it holds
  \[
    \bfcn{G}{u} \le \binom{\degree{G}{u}}{2},
  \]
  with equality if and only if there exists $v \in L \setminus \{u\}$, s.t.\
  $\neigh{G}{u} \subseteq \neigh{G}{v}$.
\end{lemma}

\rev{The meaning of this Lemma is that, when \emph{all} nodes in $R$ have degree
  \emph{at most 2}
, each node $u \in L$ can be part of \emph{at most one butterfly for
each unordered pair $(w,z)$ of $u$'s neighbors}. Specifically, if either $w$ or
$v$  has degree $1$, then there can be no butterfly involving $u$, $w$, $z$ and
another node in $L$. If both $w$ and $z$ have degree $2$, there may be at most one
butterfly involving $u$, $w$, $z$ and another node in $L$. If any node in $R$
has degree $d > 2$, a pair of neighbors of $u \in L$ could be part of more than one butterfly together with $u$ (at most $d-1$).}

We use the following technical lemma in the proof of \cref{lem:maxbfcn}.

\begin{lemma}\label{lem:techbinom}
  Let $d \in \mathbb{N}$, $d  \ge 2$. For any sequence $a_1,\dotsc,a_z$  of $1 <
  z \le d$ strictly positive naturals s.t.\ $\sum_{i=1}^z a_i \le d$, it holds
  \[
    \binom{d}{2} > \sum_{i=1}^z \binom{a_i}{2}\enspace.
  \]
\end{lemma}

\begin{proof}
  Assume by contradiction that there exists a sequence $a_1,\dotsc,a_z$  of $1 <
  z \le d$ strictly positive naturals s.t.\ $\sum_{i=1}^z a_i \le d$ and for
  which
  \begin{equation}\label{eq:techbinom1}
    \binom{d}{2} \le \sum_{i=1}^z \binom{a_i}{2}
    = \frac{1}{2} \sum_{i=1}^z (a_i^2 - a_i) \enspace.
  \end{equation}
  It holds
  \[
    2 \binom{d}{2} = d^2 - d \ge {\left( \sum_{i=1}^z a_i \right)}^2 -
    \sum_{i=1}^z a_i,
  \]
  where the inequality comes from the fact that $q^2 - q \ge g^2 - g$ for any $q
  \ge 1$, $0 \le g \le q$. Expanding the r.h.s.~of the last inequality, we
  obtain
  \begin{align}
    2 \binom{d}{2} &\ge \sum_{i=1}^z a_i^2 + 2 \sum_{i=1}^z \sum_{h=i+1}^z a_i
    a_h - \sum_{i=1}^z a_i \nonumber\\
    & = \sum_{i=1}^z (a_i^2 - a_i) + 2 \sum_{i=1}^z \sum_{h=i+1}^z a_i
    a_h \enspace. \label{eq:techbinom2}
  \end{align}
  We can combine~\eqref{eq:techbinom1} and~\eqref{eq:techbinom2} as
  \[
    \sum_{i=1}^z (a_i^2 - a_i) + 2 \sum_{i=1}^z \sum_{h=i+1}^z a_i
    a_h \le 2 \binom{d}{2} \le \sum_{i=1}^z (a_i^2 - a_i),
  \]
  which is clearly impossible because the $a_i$'s are strictly positive, thus
  the second term on the leftmost side is strictly positive. Thus we reached a
  contradiction, and the sequence $a_1,\dotsc,a_z$ cannot exists.
\end{proof}

\begin{proof}[Proof of \cref{lem:maxbfcn}]
  We start by showing that if $v$ as in the thesis exists, then $\bfcn{G}{u} =
  \binom{\degree{G}{u}}{2}$. Such a $v$ must be unique, due to the restrictions
  on the degree of the nodes in $R$. It must then hold that $\neigh{G}{u} \cap
  \neigh{G}{w} = \emptyset$ for any $w \in L \setminus \{v, u\}$, as all $a \in
  \neigh{G}{u}$ have degree exactly two and $\{(u,a), (v,a)\} \subseteq E$. From
  this fact and Fact 1, we get that $\bfcnn{G}{u}{w} = 0$ for any $w
  \in L \setminus \{v, u\}$, and that $\bfcnn{G}{u}{v} =
  \binom{\degree{G}{u}}{2}$, as $\card{\neigh{G}{u} \cap \neigh{G}{v}} =
  \card{\neigh{G}{u}} = \degree{G}{u}$ from the hypothesis. The desired equality
  follows from these facts and the definition of $\bfcn{G}{u}$
  from Equation (1) in the main text.

  We now show that if $v$ as in the thesis does \emph{not} exists, it must be
  $\bfcn{G}{u} <\binom{\degree{G}{u}}{2}$.
  Assume that there is exactly one
  $v' \in L \setminus \{u\}$ s.t.\ $\neigh{G}{u} \cap \neigh{G}{v'} \neq
  \emptyset$. Then it must be $\card{\neigh{G}{u} \cap \neigh{G}{v'}}
  < \degree{G}{u}$, otherwise $v'$ would be as $v$ in the thesis, which we
  just said can not happen. Then, from Fact 1, we get $\bfcnn{G}{u}{v'} <
  \binom{\degree{G}{u}}{2}$. From this fact, the fact that $\bfcnn{G}{u}{w} = 0$
  for any $w \in L \setminus \{v', u\}$, and Eq.\ (1) from the text, we obtain
  the desired result.

  Assume instead that there are $1 < z \le d$ distinct nodes $v_1,\dotsc,v_z \in
  L \setminus \{u\}$ s.t.\ $\neigh{G}{u} \cap \neigh{G}{v_i} \neq \emptyset$.
  The desired result $\bfcn{G}{u} < \binom{\degree{G}{u}}{2}$ follows from
  \cref{lem:techbinom} using $a_i = \card{\neigh{G}{u} \cap \neigh{G}{v_i}}$, $1
  \le i \le z$.
\end{proof}

\begin{proof}[Proof of Theorem 1]

Let $G_\mathrm{b}$ and $G_\mathrm{e}$ be the graphs output by \cref{alg:const} with inputs $s$ and $t$.
\rev{For completeness, let us recall the five properties satisfied by such graphs.
\begin{enumerate}
  \item $G_\mathrm{b}$ and $G_\mathrm{e}$ have $7+a$ left nodes (denoted with the letter $x$) and $s+t+n+2+a$ right nodes (denoted with the letter $y$); 
  \item $G_\mathrm{b}$ and $G_\mathrm{e}$ have the same left and right degree
    sequences. In particular, $x_1$ and $x_2$ have degree $s$, $x_3$ and $x_4$
    have degree $t$, $x_5$ and $x_6$ have degree $n+1$, $y_i$ has degree 2 for $1 \le i \le s+t+n+1$, and all the other left and right nodes have degree $1$;
  \item $\bfc{G_\mathrm{b}} = \bfc{G_\mathrm{e}} = n + \binom{n}{2}$;
  \item $\card{\neigh{G_\mathrm{e}}{x_5} \cap \neigh{G_\mathrm{e}}{x_6}} = n+1$, which,
  with the previous point, implies that $x_5$ and $x_6$ belong to \emph{all}
  butterflies in $G_\mathrm{e}$, and every other node $y \in R$ belongs to no
  butterfly;
  \item $\card{\neigh{G_\mathrm{b}}{x_5} \cap \neigh{G_\mathrm{b}}{x_6}} = n$,
      $\card{\neigh{G_\mathrm{b}}{x_1} \cap \neigh{G_\mathrm{b}}{x_2}} = s$,
      $\card{\neigh{G_\mathrm{b}}{x_3} \cap \neigh{G_\mathrm{b}}{x_4}} = t$,
      which implies, with point 2, that $x_5$ and $x_6$ do \emph{not} belong to
      \emph{all} butterflies in $G_\mathrm{b}$.
  \end{enumerate}
}

  For each $i \in [1,z]$, we denote with $G^i$ the graph obtained from the
  application of $\sw^{p_1}_1, \dotsc, \sw^{p_i}_i$ to $G_\mathrm{b}$. By
  definition of $q$-BSO, each $G^i$ has the same left and right degree
  sequences, and $\bfc{G^i} = \bfc{G_\mathrm{b}} = \bfc{G_\mathrm{e}}$. We set
  $G^0 = G_\mathrm{b}$.

  From properties 4 and 5 of $G_\mathrm{b}$ and $G_\mathrm{e}$ listed above,
  there must be an index $\ell \in [1,z]$ s.t.\ in the graph $G^\ell$, $x_5$ and
  $x_6$ share $n+1$ neighbors. We now show that $x_5$ and $x_6$ must share
  exactly $n$ neighbors in $G^{\ell-1}$, i.e., $\card{\neigh{G^{\ell-1}}{x_5}
  \cap \neigh{G^{\ell-1}}{x_6}} = n$.

  Assume by contradiction that $\card{\neigh{G^{\ell-1}}{x_5} \cap
  \neigh{G^{\ell-1}}{x_6}}$ $=$ $r < n$. Then, $\bfcnn{G^{\ell-1}}{x_5}{x_6} =
  \binom{r}{2}$. From Equation (2) in the main text, given the degree sequence
  of the nodes in $L$, we can write

  \begin{equation}\label{eq:maintech0}
    \bfc{G^{\ell-1}} = \sum_{i=1}^5 \sum_{j = i+1}^6
    \bfcnn{G^{\ell-1}}{x_i}{x_j} =  n + \binom{n}{2} \enspace.
  \end{equation}
  It then must be
  \begin{equation}\label{eq:maintech1}
    \sum_{i=1}^4 \sum_{j = i+1}^6 \bfcnn{G^{\ell-1}}{x_i}{x_j} = n +
    \binom{n}{2} - \binom{r}{2} \enspace.
  \end{equation}
  It holds
  \[
    \sum_{i=1}^4 \sum_{j = i+1}^6 \bfcnn{G^{\ell-1}}{x_i}{x_j}  \le
    \sum_{i=1}^4 \bfcn{G^{\ell-1}}{i},
  \]
  as some butterflies may be counted twice in the sum on the r.h.s..
  From
  \cref{lem:maxbfcn} applied to each of $x_1$, $x_2$, $x_3$, and $x_4$, it holds
  that
  \begin{equation}\label{eq:maintech2}
    \sum_{i=1}^4 \bfcn{G^{\ell-1}}{i} \le 2 \binom{s}{2} + 2 \binom{t}{2} = 2n\enspace.
  \end{equation}
  Consider for now the case $r < n-1$. If we can show that
  \begin{equation}\label{eq:maintech3}
    2n < n + \binom{n}{2} - \binom{r}{2}
  \end{equation}
  then we would have reached a contradiction, because this inequality, together
  with~\cref{eq:maintech2}, implies that~\cref{eq:maintech1} cannot be true.
  The r.h.s.\ of~\cref{eq:maintech3} decreases as $r$ increases, so if we can
  show that~\cref{eq:maintech3} holds for the maximum value of $r = n - 2$,
  then it would hold for all $r \le n - 2$. For $r = n - 2$, we can
  rewrite~\cref{eq:maintech3} as
  \begin{equation}\label{eq:maintech4}
    2n < n + \binom{n}{2} - \binom{n - 2}{2},
  \end{equation}
  which is true for any $n > 3$, i.e., for all possible values of $s$ and $t$.
  Thus we reached a contradiction and it cannot be $r < n-1$.

  Consider now the case $r = n-1$. In this
    case,
  \[
    n + \binom{n}{2} - \binom{r}{2} = 2n - 1 \enspace.
  \]
  Using this fact and~\cref{eq:maintech1}, we can write
  \begin{equation}\label{eq:maintech5}
  \begin{split}
    2n - 1 = \sum_{i=1}^3 \sum_{j=i+1}^4 \bfcnn{G^{\ell-1}}{x_i}{x_j} & +\\
    \sum_{i=1}^4 \left( \bfcnn{G^{\ell-1}}{x_i}{x_5} \right. & +
    \left. \bfcnn{G^{\ell-1}}{x_i}{x_6} \right).
  \end{split}
  \end{equation}
  Now, since $r = n - 1$ and $\degree{G^{\ell-1}}{x_5} = \degree{G^{\ell-1}}{x_6} = n
  + 1$, it must hold
  \begin{equation}\label{eq:maintech6}
    \sum_{i=1}^4 \left( \bfcnn{G^{\ell-1}}{x_i}{x_5} + \bfcnn{G^{\ell-1}}{x_i}{x_6}
    \right) \le 2
  \end{equation}
  because $x_5$ and $x_6$ can share at most two neighbors each in $R$ with one
  of $x_1$, $x_2$, $x_3$, and $x_4$. Due to the limitations on the degree of
  the nodes in $R$, if any of $x_i$, $i=5,6$ shares a neighbor with any $x_j$,
  $j=1,2,3,4$, then $x_i$ cannot share the same neighbor with any other of
  $\{x_1, x_2, x_3, x_4\} \setminus \{x_j\}$. Thus,
  combining~\cref{eq:maintech5} and~\cref{eq:maintech6}, we get that
  \begin{equation}\label{eq:maintech7}
    2n - 1 \le \sum_{i=1}^3 \sum_{j=i+1}^4
    \bfcnn{G^{\ell-1}}{x_i}{x_j} + 2 \enspace.
  \end{equation}
  It holds
  \begin{align*}
    \sum_{i=1}^3 \sum_{j=i+1}^4 \bfcnn{G^{\ell-1}}{x_i}{x_j} &= \frac{1}{2}
    \sum_{i=1}^4 \sum_{\substack{j=1\\j\neq i}}^4
    \bfcnn{G^{\ell-1}}{x_i}{x_j} \\
    & \le \frac{1}{2} \sum_{i=1}^4 \bfcn{G^{\ell-1}}{x_i} \le n,
  \end{align*}
  where the last inequality comes from~\cref{eq:maintech2}. Combining the above
  with~\cref{eq:maintech7} we obtain
  \[
    2n -1 \le n + 2
  \]
  which is only true for $n \le 3$. But from our hypothesis on $s$ and $t$, it
  must be $n > 3$, so we reached a contradiction, and it cannot be $r = n-1$.

  \noindent
  Thus, it must be $\card{\neigh{G^{\ell-1}}{x_5} \cap \neigh{G^{\ell-1}}{x_6}} =
  n$, i.e., $\bfcnn{G^{\ell-1}}{x_5}{x_6} = \binom{n}{2}$. We now show that the
  remaining $n = \binom{s}{2} + \binom{t}{2}$ butterflies in $G^{\ell-1}$ are s.t.\
  $x_1$ and $x_2$ both belong to $\binom{s}{2}$ of them, and $x_3$ and $x_4$
  both belong to $\binom{t}{2}$ of them. In other words
  $\bfcnn{G^{\ell-1}}{x_1}{x_2} = \binom{s}{2}$, $\bfcnn{G^{\ell-1}}{x_3}{x_4} =
  \binom{t}{2}$, and $\bfcnn{G^{\ell-1}}{x_i}{x_j} = 0$ for any other $(i,j) \in
  \{ (i,j) \suchthat 1 \le i \le 4, i < j \le 6\} \setminus \{(1,2),
  (3,4)\}$.

  Clearly, it must hold $\bfcnn{G^{\ell-1}}{x_i}{x_5} = \bfcnn{G^{\ell-1}}{x_i}{x_6} =
  0$ because $x_5$ and $x_6$ can each at most share one neighbor with any of
  $x_i$, $1 \le i \le 4$, which is not sufficient to obtain any butterfly to which
  both $x_i$ and $x_5$ or both $x_i$ and $x_6$ may belong. Thus, it must hold
  \[
    \sum_{i=1}^3 \sum_{j=i+1}^4 \bfcnn{G^{\ell-1}}{x_i}{x_j} = n \enspace.
  \]
  We also have
  \[
    \sum_{i=1}^3 \sum_{j=i+1}^4 \bfcnn{G^{\ell-1}}{x_i}{x_j} =  \frac{1}{2}
    \sum_{i=1}^4 \bfcn{G^{\ell-1}}{x_i} \le n,
  \]
  where the last inequality comes from~\cref{eq:maintech2}. To obtain equality,
  it must hold
  \[
    \sum_{i=1}^4 \bfcn{G^{\ell-1}}{x_i} = 2 \binom{s}{2} + 2 \binom{t}{2} \enspace.
  \]
  From \cref{lem:maxbfcn}, we have that
  \[
    \bfcn{G^{\ell-1}}{x_i} \le \left\{
      \begin{array}{ll}
        \displaystyle\binom{s}{2} &\text{if}\ i=1,2\\
        \\
        \displaystyle\binom{t}{2} &\text{if}\ i=3,4
      \end{array}
      \right.,
    \]
    with equality only if each of $x_1$, $x_2$, $x_3$, and $x_4$ shares
    \emph{all} its neighbors with another one of them. Thus, it must be that
    $x_1$ shares all its neighbors with $x_2$, and that $x_3$ shares all its
    neighbors with $x_4$, leading to the desired results that
    $\bfcnn{G^{\ell-1}}{x_1}{x_2} = \binom{s}{2}$, $\bfcnn{G^{\ell-1}}{x_3}{x_4} =
    \binom{t}{2}$, $\bfcnn{G^{\ell-1}}{x_5}{x_6} = \binom{n}{2}$, and
    $\bfcnn{G^{\ell-1}}{x_i}{x_j} = 0$ for any other $(i,j) \in \{ (i,j)
    \suchthat 1 \le i \le 5, i < j \le 6\} \setminus \{(1,2), (3,4),
    (5,6)\}$.

  At step $\ell$, the number of common neighbors between $x_5$ and $x_6$ increases
  to $n+1$, meaning that the number of butterflies involving the pair $(x_5,
  x_6)$ increases by $\binom{n+1}{2} - \binom{n}{2} = n$. Since $n =
  \binom{s}{2} + \binom{t}{2}$, the $p_\ell$-BSO $\sw^{p_\ell}_\ell \doteq (S_\ell, \sigma)$ must transform
  all the butterflies in $G^{\ell-1}$ to which $x_5$ and $x_6$ do not already
  belong, into an equal number of butterflies in $G^\ell$ to which both $x_5$
  and $x_6$ belong.

  To this end, $\sw^{p_\ell}_\ell$ must swap at least
  \begin{squishlist}
    \item $1$ edge involving $x_5$ (or $x_6$) to connect $x_5$ (or $x_6)$ to a
      node in $\neigh{G^{\ell-1}}{x_6} \setminus \neigh{G^{\ell-1}}{x_5}$  (or
      in $\neigh{G^{\ell-1}}{x_5} \setminus \neigh{G^{\ell-1}}{x_6}$): this way,
      $\card{\neigh{G^{\ell-1}}{x_5} \cap \neigh{G^{\ell-1}}{x_6}} = n +1$,
      i.e., $x_5$ and $x_6$ will share $n+1$ neighbors in $G^\ell$;
    \item $s-1$ edges involving either $x_1$ or $x_2$ to reduce the number of
      neighbors shared by these two nodes by at least $s-1$: this way
      $\card{\neigh{G^\ell}{x_1} \cap \neigh{G^\ell}{x_2}} \leq 1$, implying
      $\bfcnn{G^\ell}{x_1}{x_2} = 0$; and
    \item $t-1$ edges involving either $x_3$ or $x_4$ to reduce the number of
      neighbors shared by these two nodes by at least $t-1$: this way
      $\card{\neigh{G^\ell}{x_3} \cap \neigh{G^\ell}{x_4}} \leq 1$ and thus
      $\bfcnn{G^\ell}{x_3}{x_4} = 0$.
    \end{squishlist}

  Thus, $S_\ell \doteq \{(u_1,a_1),\dotsc,(u_{p_\ell}, a_{p_\ell})\}$ must have the following properties:
  \begin{squishlist}
  \item $S_\ell$ must contain at
    least $s-1$ edges $(u_i, a_i)$ with $u_i \in \{x_1, x_2\}$ and s.t.\ the edge
    $(u_{\sigma(i)}, a_{\sigma(i)}) \in S_\ell$ has $u_{\sigma(i)} \notin
    \{x_1,x_2\}$. Indeed if that was not the case, the edge $(u_i,
    a_{\sigma(i)})$ would already exist in $G^{\ell-1}$, because, as previously
    discussed, $x_1$ and $x_2$ share \emph{all} their neighbors in $G^{\ell-1}$.
    Thus, $S_\ell$ must contain, in addition to the $s-1$ edges as above,
    \emph{another} $s-1$ edges whose endpoint in $L$ is neither $x_1$ nor $x_2$.
    Additionally, of these $s-1$ edges, only at most $4$ may have either $x_3$
    or $x_4$ as endpoint in $L$, as if there were more, then there would be $x_i
    \in \{x_1,x_2\}$ and $x_j \in \{x_3,x_4\}$ that would share at least two
    neighbors in $G^{\ell}$, and therefore it would be
    $\bfcnn{G^{\ell}}{x_i}{x_j} > 0$, which cannot be. $S_\ell$ is only required
    to have $t-1$ edges in the form $(u_i, a_i$) with $u_i \in \{x_3,x_4\}$,
    thus really only at most $\min\{4, t-1\}$ can be as above.
  \item Similarly, $S_\ell$ must contain at least $t-1$ edges $(u_i, a_i)$ with
    $u_i \in \{x_3, x_4\}$ and s.t.\ the edge $(u_{\sigma(i)}, a_{\sigma(i)})
    \in S_\ell$ has $u_{\sigma(i)} \notin \{x_3,x_4\}$. Indeed if that was not
    the case, the edge $(u_i, a_{\sigma(i)})$ would already exist in
    $G^{\ell-1}$, because, as previously discussed, $x_3$ and $x_4$ share
    \emph{all} their neighbors in $G^{\ell-1}$; Thus, $S_\ell$ must contain, in
    addition to the $t-1$ edges as above, \emph{another} $t-1$ edges whose node
    in $L$ is neither $x_3$ nor $x_4$. Additionally, of these $t-1$ edges, only
    at most $\min\{4, t-1\}$ may have either $x_1$ or $x_2$ as endpoint in $L$,
    as if there were more, then there would be $x_i \in \{x_1,x_2\}$ and $x_j
    \in \{x_3,x_4\}$ that would share at least two neighbors in $G^{\ell}$, and
    therefore it would be $\bfcnn{G^{\ell}}{x_i}{x_j} > 0$, which cannot be.
  \end{squishlist}

  Thus, $S_\ell$ must contain at least $(2 (s-1) - \min(4, t-1)) + (2 (t-1) -
  \min(4, t-1))$ edges. For any value of $t$, this quantity is at least $2
  (s-1)$. We then have $p_\ell \ge 2 (s-1) > \rev{\bar{q}}$, where the last inequality comes
  from the definition of $s$. This fact concludes the proof.
\end{proof}

\section{Bipartite Graph Generator}\label{sec:app:algo}

We present our algorithm to generate the two bipartite graphs $G_\mathrm{b}$
and $G_\mathrm{e}$ 
used in the proof of our main result.

\begin{algorithm}[th]
  \scriptsize
  \caption{Bipartite Graph Constructor}\label{alg:const}
  \DontPrintSemicolon%
  \KwIn{Naturals \rev{$s$} and \rev{$t$} with $\rev{s \neq t} \geq 2$}
  \KwOut{Two bipartite graphs with the same degree sequences and
  $\binom{n+1}{2}$ butterflies}
  $n \gets \binom{\rev{s}}{2} + \binom{\rev{t}}{2}$; $\mathrm{add} \gets \rev{s+t}-2$\;
  \lIf{\rev{$s$} is even}{%
    $\mathrm{add} \gets \mathrm{add} + 1$
  }
  \lIf{\rev{$t$} is even}{%
    $\mathrm{add} \gets \mathrm{add} + 1$
  }
  $L \gets [x_1, \dotsc, x_{7 + \mathrm{add}}]$; $R \gets [y_1, \dotsc, y_{\rev{s+t}+n+2+\mathrm{add}}]$\;
  \tcc{butterflies btw $x_1$ and $x_2$}
  $E_\mathrm{b} \gets \{(x_1, y_l), (x_2, y_l)\}$ for $l \in [1, \rev{s}]$\;\label{line:but_i}
  \tcc{butterflies btw $x_3$ and $x_4$}
  $E_\mathrm{b} \gets E_\mathrm{b} \cup \{(x_3, y_{\rev{s}+l}), (x_4, y_{\rev{s}+l})\}$ for $l \in [1, \rev{t}]$\;\label{line:but_j}
  \tcc{butterflies btw $x_5$ and $x_6$}
  $E_\mathrm{b} \gets E_\mathrm{b} \cup \{(x_5, y_{\rev{s+t}+l}), (x_6, y_{\rev{s+t}+l})\}$ for $l \in [1, n]$\;\label{line:but_n}
  \tcc{$x_5$ and $x_6$ have deg $n+1$; $y_{\rev{s+t}+n+1}$ has deg $2$}
  $E_\mathrm{b} \gets E_\mathrm{b} \cup \{(x_5, y_{\rev{s+t}+n+1}), (x_6, y_{\rev{s+t}+n+2}), (x_7, y_{\rev{s+t}+n+1})\}$\;\label{line:add}
  \tcc{auxiliary isolated edges}
  $E_\mathrm{b} \gets E_\mathrm{b} \cup \{(x_{7+l}, y_{\rev{s+t}+n+2+l})\}$ for $l \in [1, \mathrm{add}]$\;\label{line:disc}
  \tcc{butterflies btw $x_5$ and $x_6$}
  $E_\mathrm{e} \gets \{(x_5, y_{\rev{s+t}+l}), (x_6, y_{\rev{s+t}+l})\}$ for $l \in [1, n+1]$\;\label{line:but_n1}
  \tcc{$x_1$,$x_2$,$x_3$ and $x_4$ share at most 1 neighbor}
  $E_\mathrm{e} \gets E_\mathrm{e} \cup \{(x_1, y_1), (x_2, y_1), (x_3, y_{\rev{s}+1}), (x_4, y_{\rev{s}+1})\}$\;\label{line:common}
  $h_1 \gets \left\lfloor(\rev{s}-1)/2\right\rfloor$; $h_2 \gets \left\lfloor(\rev{t}-1)/2\right\rfloor$\;
  \tcc{we split $\rev{s}-1$ right nodes btw $x_1$ and $x_2$}
  $E_\mathrm{e} \gets E_\mathrm{e} \cup \{(x_1, y_{1+l}), (x_2, y_{1+h_1+l})\}$ for $l \in [1, h_1]$\;\label{line:half1}
  \tcc{we split $\rev{t}-1$ right nodes btw $x_3$ and $x_4$}
  $E_\mathrm{e} \gets E_\mathrm{e} \cup \{(x_3, y_{\rev{s}+1+l}), (x_4, y_{\rev{s}+1+h_2+l})\}$ for $l \in [1, h_2]$\;\label{line:half2}
  \tcc{$x_1$, $x_2$ have deg $\rev{s}$ and $x_3$, $x_4$ have deg $\rev{t}$}
  $a_1 \gets \rev{s} - (1 + h_1)$; $a_2 \gets \rev{t} - (1 + h_2)$\;
  $E_\mathrm{e} \gets E_\mathrm{e} \cup \{(x_1, y_{\rev{s+t}+n+2+l}), (x_2, y_{\rev{s+t}+n+2+a_1+l})\}$ for $l \in [1, a_1]$\;\label{line:ensurei}
  $E_\mathrm{e} \gets E_\mathrm{e} \cup \{(x_3, y_{\rev{s+t}+n+2+2a_1+l}), (x_4, y_{\rev{s+t}+n+2+2a_1+a_2+l})\}$ for $l \in [1, a_2]$\;\label{line:ensurej}
  \tcc{$x_7$ and $y_{\rev{s+t}+n+2}$ have deg $1$}
  $E_\mathrm{e} \gets E_\mathrm{e} \cup \{(x_7, y_{\rev{s+t}+n+2})\}$\;\label{line:x7}
  \tcc{the first $(\rev{s+t})$ right nodes have deg $2$}
  $E_\mathrm{e} \gets E_\mathrm{e} \cup \{(x_{7+l}, y_{1+l})\}$ for $l \in [1, \rev{s}-1]$\;\label{line:ensureright_s}
  $E_\mathrm{e} \gets E_\mathrm{e} \cup \{(x_{7+\rev{s}-1+l}, y_{\rev{s}+1+l})\}$ for $l \in [1, \rev{t}-1]$\;
  $\mathrm{add} \gets 0$\;
  \lIf{\rev{$s$} is even}{%
    $E_\mathrm{e} \gets E_\mathrm{e} \cup \{(x_{7+\rev{s+t}-1}, y_\rev{s})\}$; $\mathrm{add} \gets \mathrm{add} + 1$
  }
  \lIf{\rev{$t$} is even}{%
    $E_\mathrm{e} \gets E_\mathrm{e} \cup \{(x_{7+\rev{s+t}-1+\mathrm{add}}, y_{\rev{s+t}})\}$
  }\label{line:ensureright_e}
  $G_\mathrm{b} \gets (L, R, E_\mathrm{b})$; $G_\mathrm{e} \gets (L, R, E_\mathrm{e})$\;
  \Return{$G_\mathrm{b}$, $G_\mathrm{e}$}\;
\end{algorithm}

The algorithm (pseudocode in \cref{alg:const}) receives in input two naturals $\rev{s \neq t} \in \mathbb{N}\, \wedge \, \rev{s,t} \geq 2$, and generates two bipartite
$G_\mathrm{b} = (L, R, E_\mathrm{b})$ and $G_\mathrm{e} = (L, R, E_\mathrm{e})$
with the same left and right degree sequence and each with $\binom{n+1}{2}$ butterflies, where
$n = \binom{\rev{s}}{2} + \binom{\rev{t}}{2}$.

The algorithm starts with the creation of the set of left nodes $L$ (any node in this set will be denoted as $x_w$, for some $w$) and of right nodes $R$ (any
node in this set will be denoted as $y_w$, for some $w$), equal for both graphs.
Then, it populates the edge set $E_\mathrm{b}$ of $G_\mathrm{b}$ and the edge set $E_\mathrm{e}$ of $G_\mathrm{e}$.
In $G_\mathrm{b}$, the butterflies involve three pairs of left nodes: $(x_1, x_2)$, $(x_3, x_4)$, and $(x_5, x_6)$.
Nodes $x_1$ and $x_2$ have degree \rev{$s$} and share \rev{$s$} neighbors (line~\ref{line:but_i});
nodes $x_3$ and $x_4$ have degree \rev{$t$} and share \rev{$t$} neighbors (line~\ref{line:but_j});
and nodes $x_5$ and $x_6$ have degree $n+1$ and share $n$ neighbors (line~\ref{line:but_n}).
In $G_\mathrm{e}$, the butterflies involve only the left nodes $x_5$ and $x_6$, which share $n+1$ neighbors (line~\ref{line:but_n1}).

We will construct the two edge sets so that any other pair of left nodes share at most one neighbor.
Thus, it hols that $\bfc{G_\mathrm{b}} = \binom{\rev{s}}{2} + \binom{\rev{t}}{2} + \binom{n}{2} = \binom{n+1}{2} = \bfc{G_\mathrm{e}}$.
Nodes $x_5$ and $x_6$ have degree $n+1$ in both graphs, but, so far, in $G_\mathrm{b}$ we have added only $n$ edges to them.
Similarly, node $y_{\rev{s+t}+n+1}$ have degree $2$ in $G_\mathrm{e}$, but, so far, in $G_\mathrm{b}$ it is connected to only one left node.
Therefore, we insert in $E_\mathrm{b}$ one edge involving $x_5$, one edge involving $x_6$ (we connect them to different right nodes to avoid creating other butterflies between them), and one edge involving $y_{\rev{s+t}+n+1}$ (line~\ref{line:add}).
The construction of $E_\mathrm{b}$ is completed the addition of up to
\rev{$s+t$} isolated edges, i.e., edges not connected with any other edge (line~\ref{line:disc}).
These edges do not participate in any butterfly, because they involve nodes with degree $1$.
They are needed to guarantee that $G_\mathrm{b}$ and $G_\mathrm{e}$ have the
same degree sequences, as it will become evident as we outline the other edges included in $E_\mathrm{e}$. 

We add edges to $E_\mathrm{e}$ in such a way that the pairs of vertices $(x_1, x_2)$ and $(x_3, x_4)$ have no more than one neighbor in common, so they do not belong to any butterfly. 
The idea is to \textbf{(i)} connect both $x_1$ and $x_2$ to $y_1$ (line~\ref{line:common}), \textbf{(ii)} divide the remaining nodes to which $x_1$ is connected
in $G_\mathrm{b}$ into two groups $(y_2, \dotsc, y_{h_1})$ and $(y_{h_1+1}, \dotsc, y_{2h_1 + 1})$ with $h_1 = \left\lfloor(\rev{s}-1)/2\right\rfloor$, and \textbf{(iii)} insert into $E_\mathrm{e}$ one edge between $x_1$ and each node in the first group, and one edge between $x_2$ and each node in the second group
(line~\ref{line:half1}). 
Similarly, to ensure that $x_3$ and $x_4$ only have one neighbor in common, we \textbf{(i)} connect both $x_3$ and $x_4$ to $y_{\rev{s}+1}$ (line~\ref{line:common}), \textbf{(ii)} divide the remaining nodes to which $x_3$ is connected
in $G_\mathrm{b}$ into two groups $(y_{\rev{s}+2}, \dotsc, y_{\rev{s}+2+h_2})$ and $(y_{\rev{s}+3+h_2}, \dotsc, y_{\rev{s}+1+2h_2})$ with $h_2 = \left\lfloor(\rev{t}-1)/2\right\rfloor$, and \textbf{(iii)} insert into $E_\mathrm{e}$ one edge between $x_3$ and each node in the first group, and one edge between $x_4$ and each node in the second group (line~\ref{line:half2}). 
So far, in $E_\mathrm{e}$ we have added only $(1 + h_1)$ of the \rev{$s$} neighbors that $x_1$ and $x_2$ must have, and only $(1 + h_2)$ of the \rev{$t$} neighbors that $x_3$ and $x_4$ must have.
Thus, we include $\rev{s} - (1 + h_1)$ edges to different right nodes (to avoid creating butterflies between them) for $x_1$ and $x_2$ (line~\ref{line:ensurei}), and $\rev{t} - (1 + h_2)$ edges to different right nodes for $x_3$ and $x_4$ 
(line~\ref{line:ensurej}).
Similarly, so far we have added only one edge to the right nodes $y_2, \dotsc, y_{\rev{s}-1}, y_{\rev{s}+2}, \dotsc, y_{\rev{s+t}-1}$, and up to one edge to the right nodes $y_\rev{s}$ and $y_{\rev{s+t}}$. 
In fact, depending on the value of $h_1$ and $h_2$, such nodes are included/excluded from the group of nodes connected to $x_2$ and $x_4$, respectively. 
Since all these nodes have degree $2$ in $G_\mathrm{b}$, we add the missing edges to $E_\mathrm{e}$ (lines~\ref{line:ensureright_s}--\ref{line:ensureright_e}).
Lastly, nodes $x_7$ and $y_{\rev{s+t}+n+2}$ have degree $1$ in $G_\mathrm{b}$, and thus they must be connected to one neighbor also in $G_\mathrm{e}$ (line~\ref{line:x7}).
Finally, the two
graphs $G_\mathrm{b} \doteq (L, R, E_\mathrm{b})$ and $G_\mathrm{e} \doteq (L,
R, E_\mathrm{e})$ are returned.

\smallskip

\section{Connecting Arbitrary Bipartite Graphs in the State Space via $q$-BSO}\label{sec:app:q5}

This section shows how to build a $q$-BSO that transforms an arbitrary bipartite graph $G' \doteq (L, R, E')$ into another arbitrary graph $G'' \doteq (L, R, E'')$ with the same left and right degree sequences and number of butterflies. 
The $q$-BSO needs to swap all links except those that are in common between the two graphs, i.e., $q = \lvert E' \setminus E'' \rvert$.
Let $E' \cap E'' = \mylist{e'_1,\dotsc,e'_c} = \mylist{e''_1,\dotsc,e''_c}$ be the list of $c$ edges in common to the two graphs, $E'=\mylist{e'_1,\dotsc,e'_{\card{E'}}}$ the list of edges of $G'$, $E''=\mylist{e''_1,\dotsc,e''_{\card{E'}}}$ the list of edges of $G''$, and $q = \card{E'} - c$ the number of edges unique to $G'$. 
  We now show a $q$-BSO operation $\sw \doteq (\mylist{e'_{c+1},\dotsc,e'_{|E'|}}, \sigma)$ that
  transforms $G'$ into $G''$. We build the derangement $\sigma$ incrementally, denoting with
  $\mathrm{\hat{Im}}(\sigma)$ the \emph{current} (i.e., as $\sigma$ is being
  built) image of $\sigma$. At the beginning of the construction process,
  $\mathrm{\hat{Im}}(\sigma) = \emptyset$, and at the end,
  $\mathrm{\hat{Im}}(\sigma) = \mathrm{Im}(\sigma) =
  \{1, \dotsc, q\}$. 
  For each $i \in \{1, \dotsc, q\}$, let $e'_{c+i} \doteq (u, v)$ and $k \in \{1, \dotsc, q\}$ be any index such that $k \notin \mathrm{\hat{Im}}(\sigma)$ and $e''_{c+k} \doteq (u, w)$. Then, we set $\sigma(i)
  \doteq k$. The index $k$ must always exist because $G'$ and $G''$ have the
  same degree sequences.

\end{document}